\documentclass[11pt]{article}

\usepackage[margin=1in]{geometry}
\usepackage[T1]{fontenc}
\usepackage[utf8]{inputenc}
\usepackage{lmodern}
\usepackage{microtype}
\usepackage{amsmath,amssymb,amsthm}
\usepackage{booktabs}
\usepackage{xcolor}
\usepackage{graphicx}
\usepackage{float}
\usepackage[hidelinks]{hyperref}
\hypersetup{
  pdftitle={Fresh-Challenge VDF Attestations for Model-Relative Response Latency},
  pdfauthor={Ansar Yesmukhanov and Aruzhan Tlessova}
}

\newtheorem{definition}{Definition}

\newtheorem{lemma}{Lemma}
\newtheorem{claim}{Claim}
\newtheorem{proposition}{Proposition}

\newcommand{\Verify}{\operatorname{Verify}}
\newcommand{\Eval}{\operatorname{Eval}}
\newcommand{\Setup}{\operatorname{Setup}}

\title{Fresh-Challenge VDF Attestations for\newline Model-Relative Response Latency}
\author{Ansar Yesmukhanov \and Aruzhan Tlessova\\[0.5em]\small Nazarbayev University, Astana, Kazakhstan}
\date{}

\begin{document}
\maketitle

\begin{abstract}
Can a finite verifier obtain public, model-relative evidence about response latency for sequential computation? Verifiable delay functions (VDFs) make this possible in principle: evaluation requires $T$ sequential steps, whereas verification is efficient in the security parameter and polylogarithmic in the numerical value of $T$ for standard constructions. Thus a delay can be astronomically large to evaluate yet succinctly represented and feasibly checked. A VDF proof for a chosen message alone is insufficient because it may be precomputed. We specify and analyze \emph{Fresh-Challenge VDF Attestations} (FCLA), a protocol composition that binds a VDF to an unpredictable public challenge, a message, and independently auditable release and receipt records. Under explicit assumptions about VDF sequentiality, the challenge source, witness logs, and a calibrated upper bound on an adversary's sequential evaluation rate, an accepted FCLA transcript is inconsistent with post-challenge generation by an adversary in that bounded model. The result neither identifies a named claimant nor excludes relaying, outsourcing, or a faster unmodeled machine. A benchmark of the public reference implementation confirms the expected empirical separation between evaluation and verification on one documented machine. Our contribution is a protocol/design analysis and benchmarked reference implementation layer, not a new VDF construction or cryptographic primitive.
\end{abstract}

\section{Introduction}

Cryptographic protocols routinely turn computational asymmetry into public evidence. Proof-of-work makes resource expenditure costly; time-lock puzzles delay access to a value; and verifiable delay functions (VDFs) provide a particularly sharp notion of sequential work: evaluation takes a prescribed number of inherently sequential steps while verification is efficient and public \cite{dworknaor,rsw96,boneh2018vdf,wesolowski2019}. The asymmetry invites an unusual but precise question: can a public transcript demonstrate that a response could not have been produced by a bounded physical evaluator after a challenge became known?

The answer can be yes, conditionally. Suppose a challenge is released at time $t_0$, a response arrives at time $t_1$, and a valid proof requires $T$ sequential steps after the challenge is known. If a declared evaluator model permits at most $\rho$ such steps per second and $t_1-t_0<T/\rho$, then an accepted response is inconsistent with post-challenge generation by that evaluator model. The verifier need not repeat the astronomical evaluation. This is a model-relative latency claim, not a claim that a named claimant performed the computation.

This paper identifies both the strongest defensible conclusion and its boundary. A valid VDF transcript certifies a relation between public inputs and outputs under VDF security assumptions. Without a fresh challenge and a deadline it does not establish that the computation happened during any particular interval. Without an identity-bearing secret or trusted external binding, it does not identify its producer. These are not weaknesses of the construction; they specify exactly what the evidence means.

Our contribution is a protocol and analysis rather than a new VDF construction:

\begin{itemize}
  \item We give a protocol/design analysis of a model-relative response-latency claim, including the relationship among delay $T$, a calibrated sequential-rate bound $\rho$, and an auditable response window.
  \item We show why a message-bound VDF is not itself a response-speed proof, then specify FCLA to prevent precomputation using a fresh public challenge.
  \item We state an informal FCLA security proposition and an identity-blindness observation that make the protocol's attribution and threat-model boundary explicit.
  \item We analyze astronomical delays, including why large $T$ values remain compact to represent and feasible to verify in principle, and document the representation limits of a practical implementation.
  \item We benchmark the public reference implementation over six delays, reporting evaluation and verification times, measurement variance, and the observed sequential iteration rate on a documented machine.
\end{itemize}

The intent is constructive. A public board that verifies VDFs can be an engaging implementation and useful educational artifact. To become a research claim, it should describe itself as an implementation of a standard primitive and make the protocol's inference boundary explicit.

\section{Background and model}

\subsection{Verifiable delay functions}

We use the standard informal interface of a VDF \cite{boneh2018vdf,wesolowski2019}:
\[
  \Setup(1^\lambda,T)\rightarrow pp,\qquad
  \Eval(pp,x,T)\rightarrow (y,\pi),\qquad
  \Verify(pp,x,T,y,\pi)\in\{0,1\}.
\]
Correctness requires that honestly evaluated outputs verify. Sequentiality informally says that, after $x$ becomes available, no efficient adversary can produce $(y,\pi)$ substantially faster in \emph{sequential depth} than the prescribed delay. Efficient public verification means that checking need not repeat the $T$ sequential steps. Concrete constructions require additional assumptions about the underlying group and setup; class groups of imaginary quadratic fields and groups of unknown order are common choices \cite{wesolowski2019}.

For a Wesolowski-style construction, the required properties are correctness, proof soundness, uniqueness (where supplied by the construction), and sequentiality in the chosen group of unknown order. These properties depend on the selected setup and parameter-generation method: a deployment must specify the group, its security rationale, whether a trusted setup exists, and why no party retained a trapdoor that defeats the intended assumption \cite{boneh2018vdf,wesolowski2019}. The protocol does not claim to improve any of these VDF properties.

Sequentiality is not a statement about universal wall-clock time. Let $\rho$ be an upper bound on the number of modeled sequential VDF steps per second available to an adversary. A delay $T$ yields the conditional lower bound $T/\rho$ seconds. The value of $\rho$ is not supplied by cryptography: it is an engineering, hardware, and deployment assumption.

\subsection{Astronomical delays remain describable}

The evaluator performs $T$ sequential steps, but the verifier does not expand $T$ into a list of steps. In a Wesolowski-style VDF, verification consists of a small number of group operations together with exponent arithmetic derived from $T$; computing quantities such as $2^T \bmod \ell$ is polylogarithmic in the numerical value of $T$, hence polynomial in its bit-length $\log T$, subject to the construction and group-operation costs \cite{wesolowski2019}. The proof itself has constant-size group elements. Thus a delay such as $T=10^{100}$ has a representation of only about 333 bits even though direct evaluation is fantastically large.

This asymmetry is what makes the proposed claim meaningful. Given a public challenge released at $t_0$ and a witnessed receipt at $t_1$, select $T$ so that
\[
  T > \rho(t_1-t_0+\epsilon),
\]
where $\epsilon$ covers clock uncertainty and network delay. A valid proof then conflicts with the declared evaluator bound, subject to the assumptions in Section~5.1. The statement is deliberately finite and falsifiable: it need not invoke an undefined upper bound on all possible machines.

There is an important implementation distinction. The current VDP demonstration represents $T$ as Rust \texttt{u64}, whose maximum is $2^{64}-1\approx1.84\times10^{19}$. It can demonstrate long delays, but it cannot encode arbitrary astronomical values. A deployment that needs delays beyond this range must use a VDF implementation whose iteration-count interface, serialization, and verifier all accept a big-integer $T$. This is an engineering extension, not a change to the FCLA inference rule.

\subsection{Four claims that must be kept separate}

\begin{center}
\begin{tabular}{p{0.25\linewidth}p{0.66\linewidth}}
\toprule
Claim & What is required \\
\midrule
Correct evaluation & A valid VDF proof under public parameters. \\
Sequential work & A VDF sequentiality assumption, a fixed delay $T$, and an input unavailable before the claimed evaluation. \\
Response latency & Sequential work plus an unpredictable challenge and trustworthy evidence of challenge release and response receipt. \\
Identity or authorship & An identity-binding mechanism, such as a signature key, credential, or trusted registration process. \\
\bottomrule
\end{tabular}
\end{center}

The last row is not a defect in VDFs. It is a category distinction. Authentication tells a verifier \emph{which principal} made a statement. A VDF has no secret tied to a principal and therefore cannot furnish that binding by itself.

\section{Why a capability-only VDF transcript cannot authenticate}

\begin{definition}[Capability-only transcript]
A transcript $\tau$ is capability-only if acceptance is decided by a public predicate $V(\tau)$ and $\tau$ contains no value whose production or possession is uniquely bound to a named principal. In particular, a message $m$, public parameters $pp$, a delay $T$, and a VDF output $(y,\pi)$ are capability-only.
\end{definition}

\begin{lemma}[Identity blindness]
Let $V$ be any public verifier for capability-only transcripts. Let $P_0$ and $P_1$ be two producers whose output distributions over accepted transcripts are identical. For every (possibly randomized) verifier $A$ that receives only an accepted transcript,
\[
\Pr[A(\tau)=0\mid \tau\leftarrow P_0]
=
\Pr[A(\tau)=0\mid \tau\leftarrow P_1].
\]
Consequently, no such verifier can authenticate whether $P_0$ or $P_1$ produced $\tau$ with advantage over random guessing.
\end{lemma}

\begin{proof}
The view of $A$ is a sample from the same distribution in both experiments. Applying the same randomized algorithm $A$ to identically distributed inputs yields identical output distributions. Therefore its distinguishing advantage is zero.
\end{proof}

This elementary observation is operationally important. A VDF proof may be difficult to generate, yet difficulty does not make a transcript identify its generator. If a normal producer can compute a proof eventually, its final valid proof has the same public verification behavior as the proof of a producer with a faster machine. Any identity inference therefore enters through an assumption outside the VDF transcript: a registered public key, a trusted observer, a physical measurement, or a prior agreement.

\section{The precomputation problem}

Consider the natural construction
\[
  x=H(m),\qquad (y,\pi)\leftarrow \Eval(pp,x,T),\qquad \tau=(m,T,y,\pi).
\]
It binds a proof to $m$ in the narrow sense that a proof for one message will not verify for a different message. It does not bind the computation to the time at which $m$ was displayed or submitted.

\begin{claim}[Offline precomputation]
If an adversary chooses $m$ before an alleged interaction, then it can compute $(y,\pi)$ before the interaction and release $\tau$ at any later time. The distribution of $\tau$ is the same as if evaluation began at release time.
\end{claim}

\begin{proof}
The VDF input is the deterministic value $H(m)$ and has no dependence on an event occurring at release time. Evaluation can therefore be run as soon as $m$ is chosen. The public verifier receives the same tuple regardless of when the evaluator started.
\end{proof}

Appending a server timestamp after the fact does not repair this attack: it timestamps publication, not the beginning of computation. The input must include an unpredictable value that was unavailable before a public release event.

\section{Fresh Challenge Latency Attestation}

FCLA gives the minimal structure for a meaningful, limited claim. It is designed for a setting in which a claimant wishes to demonstrate that it answered a fresh public challenge more quickly than a stated conventional bound.

\subsection{Entities and assumptions}

The protocol has a public VDF parameter set $pp$, a collision-resistant hash function $H$, a public challenge source $B$, and one or more append-only witness logs $L_1,\ldots,L_k$. The source $B$ releases a value $r$ that is unpredictable before its release time $t_0$. A witness log records an authenticated receipt time $t_1$. The claimant may optionally attach a digital signature to associate the response with a public key; this optional addition proves control of a key, not the physical identity of a human or agent.

The protocol assumes: (i) correctness, proof soundness, and sequentiality for the selected VDF construction and parameters; (ii) unpredictability of $r$ before $t_0$; (iii) integrity and append-only availability of the witness records; (iv) collision resistance of $H$; and (v) a reproducibly calibrated upper rate $\rho$ for the modeled adversary. Removing any assumption weakens the inference correspondingly.

The trust model is substantive. If a beacon leaks $r$ early, is predictable, or equivocates by showing different challenges to different parties, precomputation may reappear. If witnesses collude with a claimant, backdate a receipt, equivocate about a log, or use manipulable clocks, $t_1-t_0$ is not trustworthy. ``Independent witnesses'' means that their operators, clock sources, network paths, and log-signing keys should not share a failure mode that permits a single actor to fabricate the required release and receipt history. A practical deployment should publish signed records, log consistency proofs, clock synchronization methods, and a threshold rule defining how many mutually independent records are required.

\subsection{Protocol}

\begin{enumerate}
  \item \textbf{Challenge release.} The beacon publishes a signed or otherwise independently auditable record $e_0=(\mathsf{sid},r,t_0)$, where $\mathsf{sid}$ is a unique session identifier.
  \item \textbf{Input derivation.} For the claimed message $m$ and delay $T$, compute
  \[
     x=H(\textsf{FCLA-v1}\parallel\mathsf{sid}\parallel r\parallel H(m)\parallel T\parallel \mathsf{ctx}),
  \]
  where $\mathsf{ctx}$ fixes the VDF parameters and application context.
  \item \textbf{Evaluation.} The claimant computes $(y,\pi)\leftarrow\Eval(pp,x,T)$ and sends $e_1=(\mathsf{sid},m,T,y,\pi)$ to the witnesses.
  \item \textbf{Receipt.} Each witness appends an authenticated receipt record $(H(e_1),t_1)$ to its public log.
  \item \textbf{Verification.} A verifier checks the beacon record, every selected witness receipt, recomputes $x$, verifies $\Verify(pp,x,T,y,\pi)=1$, and checks the stated timing rule $t_1-t_0<\Delta$.
\end{enumerate}

Including a domain separator, session identifier, parameter context, and message hash prevents accidental cross-protocol reuse. Multiple independent beacons or logs can reduce reliance on any single time authority, although they cannot eliminate the need for an assumption about recording events.

\subsection{Conditional latency proposition}

\begin{proposition}[FCLA latency claim, informal]
Suppose $r$ is unavailable to an adversary before $t_0$, the beacon and selected witness records are authentic and non-equivocating, $H$ is collision resistant, and the VDF satisfies proof soundness and sequentiality against adversaries whose sequential evaluation rate is at most $\rho$. If FCLA accepts a response with
\[
 t_1-t_0 < T/\rho,
\]
then, except with negligible probability, the response could not have been generated after release of $r$ by an adversary within that modeled rate bound. The conclusion does not attribute generation to a named claimant and does not exclude relaying, outsourced computation, or a faster machine outside the model.
\end{proposition}

\begin{proof}[Proof sketch]
The verified transcript fixes $x$ to the unpredictable value $r$, the session, the message, delay, and context. Before $t_0$, the adversary cannot know $x$ except with negligible probability under the beacon-unpredictability assumption. After $t_0$, producing an accepted pair for $x$ in fewer than $T/\rho$ seconds contradicts VDF sequentiality at the assumed rate, unless the adversary finds a hash collision or falsifies a timing record. Each alternative is excluded by an assumption.
\end{proof}

The conclusion is intentionally conditional. It does \emph{not} say that a claimant performed the computation, that no faster machine exists, or that the VDF was completed in less than an objective physical time. It says only that the recorded result is inconsistent with a specified adversary model.

\subsection{Calibrating the evaluator model}

The rate bound $\rho$ is not an implementation detail: it is the bridge between a cryptographic sequentiality statement and a real-world latency claim. Before a deployment uses FCLA to exclude a model, it should publish a reproducible calibration protocol. At minimum, this protocol should specify the exact VDF implementation and parameters; the hardware class, operating system, compiler, and optimization flags; all permitted specialized hardware; the permitted parallelism and algorithmic optimizations; the benchmark inputs and number of trials; and a conservative statistical rule for converting measured step rates into an upper bound. It must also budget network delay, witness processing delay, clock synchronization error, and clock uncertainty in a publicly stated margin $\epsilon$.

Calibration must be adversary-aware. Benchmarking an ordinary laptop does not justify a bound against GPUs, FPGAs, ASICs, cloud clusters, optimized implementations, or outsourced evaluators unless those possibilities are expressly excluded by the evaluator model. The timing rule should therefore be stated as, for example,
\[
  t_1-t_0+\epsilon < T/\rho,
\]
where the published $\rho$ and $\epsilon$ are selected before the challenge and reproducible by third parties. This makes an FCLA claim falsifiable: anyone may challenge the stated rate or timing margin with a better measurement or a faster evaluator within the declared model.

\section{Experimental evaluation}

We measured the public VDP implementation at commit \texttt{bb111dc} (the full identifier is supplied with the ancillary benchmark files), which invokes the \texttt{vdf} crate's 512-bit Wesolowski parameters. The experiment measures the behavior of this particular implementation and configuration; it neither experimentally proves VDF sequentiality nor establishes a trusted beacon or witness service. In particular, the results below do not justify a rate bound for GPUs, FPGAs, ASICs, cloud systems, outsourced evaluators, or any machine outside the stated model.

\begin{center}
\small
\begin{tabular}{ll}
\toprule
Item & Configuration \\
\midrule
Machine & Apple M4 system; 10 logical cores; 24 GB memory \\
Operating system & macOS 15.5 \\
Software & Rust 1.96.0; Cargo release build; VDP commit \texttt{bb111dc} \\
VDF & Wesolowski VDF, 512-bit parameters; fixed public byte-string input \\
Protocol & One unreported warm-up; 10 timed trials per phase and delay \\
Timing & Rust \texttt{Instant} monotonic clock; end-to-end \texttt{solve}/\texttt{verify} calls \\
\bottomrule
\end{tabular}
\end{center}

Before the full run, a three-trial calibration at $T=10^4$ gave a median evaluation time below one second. We therefore selected six approximately logarithmically spaced values from $10^2$ through $10^6$, keeping the largest median evaluation well below 60 seconds. For every timed evaluation, the resulting proof was verified successfully; for each delay, a proof was also checked against an altered input and rejected. Table~\ref{tab:benchmark} reports medians and sample standard deviations across the 10 timed trials.

\begin{table}[H]
\centering
\small
\caption{Measured VDP timings on the stated Apple M4 configuration. Times are median $\pm$ sample standard deviation in milliseconds; the rate is $T$ divided by median evaluation time.}
\label{tab:benchmark}
\begin{tabular}{r r r r}
\toprule
$T$ & Evaluation (ms) & Verification (ms) & Iterations/s \\
\midrule
100       & $10.024 \pm 0.075$ & $9.800 \pm 0.057$  & 9,976 \\
631       & $15.472 \pm 0.657$ & $12.297 \pm 1.156$ & 40,784 \\
3,981     & $41.580 \pm 0.580$ & $12.235 \pm 0.163$ & 95,743 \\
25,119    & $199.303 \pm 3.416$ & $12.009 \pm 0.372$ & 126,034 \\
158,489   & $1,140.304 \pm 24.371$ & $11.780 \pm 0.110$ & 138,988 \\
1,000,000 & $7,021.828 \pm 117.587$ & $11.691 \pm 0.160$ & 142,413 \\
\bottomrule
\end{tabular}
\end{table}

\begin{figure}[H]
\centering
\includegraphics[width=0.88\linewidth]{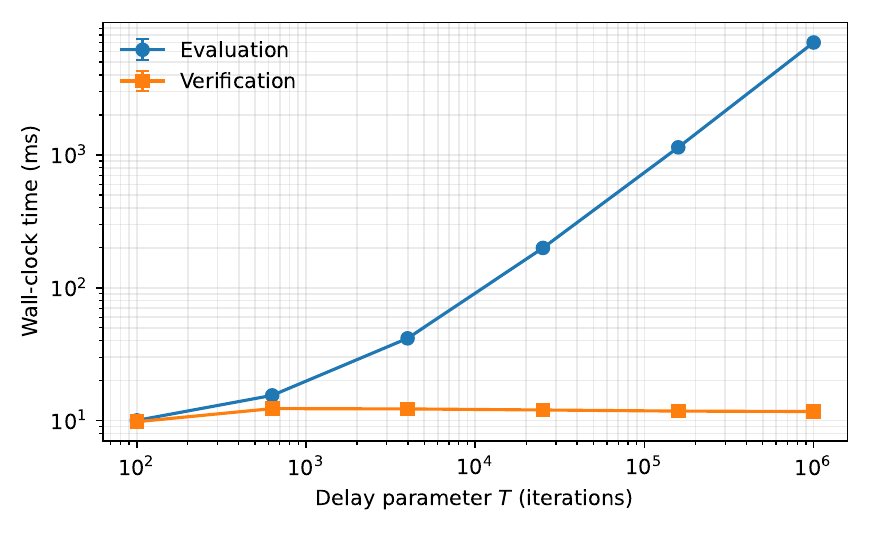}
\caption{Evaluation and verification wall-clock time versus delay. Both axes are logarithmic; error bars show sample standard deviations over 10 trials.}
\label{fig:benchmark}
\end{figure}

Figure~\ref{fig:benchmark} makes the intended asymmetry visible: evaluation increases substantially with $T$, while verification remains close to 12 ms throughout the larger-delay range. The small-$T$ measurements include a fixed cost that explains the lower apparent evaluation rate at the first points; at $T=10^6$, the observed rate is approximately $1.42\times10^5$ iterations per second on this machine. This number is descriptive, not a universal $\rho$. A deployment must choose a conservative rate bound for its declared adversary class before challenge release, together with its timing margin $\epsilon$.

\section{Design consequences for public VDF boards}

An implementation that accepts $(m,T,y,\pi)$ can accurately advertise ``public verification of message-bound sequential work.'' It should not advertise ``authorship'' or ``proof of a being's nature.'' To support FCLA, a deployment must add:

\begin{itemize}
  \item a documented source of unpredictable public challenges;
  \item durable, independently auditable records for release and receipt events, including a threshold rule and log-consistency evidence for multiple witnesses;
  \item a parameter document specifying the VDF group, security level, setup and parameter-generation assumptions, and the calibration procedure for $\rho$;
  \item canonical byte encoding and domain separation for all signed, hashed, and VDF-derived values;
  \item replay protection using a unique session identifier; and
  \item an explicit statement that optional signatures bind a response to a key, while identity remains a social or institutional assertion; and
  \item a threat-model document covering beacon leakage or equivocation, witness collusion, clock manipulation, network delay, relaying, outsourcing, and unmodeled accelerators.
\end{itemize}

The choice of $T$ should be made from a target latency claim, not a rhetorical tier. For example, a target $\Delta$ requires $T>\rho\Delta$ with a conservative estimate of $\rho$ and a margin for clock error, network delay, and benchmark variance. A value that takes hours on a current laptop is evidence of work only if the protocol prevents offline preparation and specifies which class of evaluators is being ruled out.

\section{Limitations and scope}

Several limitations are fundamental. First, a VDF's sequentiality is conditional on the chosen construction, setup assumptions, computational model, and the absence of an undiscovered algorithmic shortcut. Second, external timestamps establish only the trust properties of the systems recording them. Third, a relay, outsourced evaluator, or faster unmodeled machine is consistent with an accepted transcript. Fourth, a highly capable but modeled adversary can always be accommodated by raising $\rho$; no finite measurement rules out all conceivable machines. Finally, even a response that exceeds a stated bound is evidence about a transcript and model, not a conclusion about consciousness, intention, divinity, or any other metaphysical predicate.

These limitations do not make the protocol useless. They direct it toward applications where model-relative evidence is valuable: public challenge benchmarks, transparent claims about specialized hardware, timed competitions, and auditable service-level demonstrations. For those applications, stating the limits precisely is a feature, not a concession.

\section{Conclusion}

VDFs provide publicly verifiable sequential-work evidence, not identity or metaphysical authentication. We used the precomputation attack and an identity-blindness observation to delimit that evidence, then specified FCLA as a composition of a VDF, fresh challenge, and auditable timing records. Its strongest practical inference is deliberately narrow: a fresh, publicly witnessed response can be inconsistent with an explicitly calibrated class of evaluators. This gives a rigorous foundation for a public VDF board while avoiding claims that the underlying cryptography cannot justify.

\section*{Acknowledgements}

The authors thank Rustem Takhanov for his helpful feedback on the manuscript and for suggesting the experimental evaluation.

\paragraph{Reproducibility.} A reference implementation of the message-bound VDF demonstration is available at \url{https://github.com/AnsarYesma/vdp}. The ancillary \texttt{experiments/} directory supplies the benchmark source, raw per-trial CSV data, analysis script, summary data, and plot used above. We also implemented and tested a minimal Rust FCLA transcript-binding layer against that codebase; it uses a domain-separated, length-delimited SHA-256 input, validates deadline ordering, and verifies that a proof fails when the message changes. The benchmark is not a deployment or an FCLA latency-attestation experiment: a complete FCLA service must still verify external beacon and witness-log records before making any empirical latency claim.

\bibliographystyle{plain}
\bibliography{references}

\end{document}